\newif\ifprocs
\procstrue
\procsfalse %last one wins (comment this line to get LIPIcs proceedings format)

\ifprocs
\documentclass[a4paper,USenglish,numberwithinsect]{lipics-v2019}
\usepackage{babel}
\else
\documentclass[11pt]{article}
\usepackage[margin=1.0in]{geometry}
\usepackage[utf8]{inputenc}
\usepackage[english]{babel}
\usepackage{subfig}
\fi

\usepackage{amssymb}
\usepackage{amsfonts}
\usepackage{amsmath}
\usepackage{amsthm}
\usepackage{mathtools} % for ceiling function
\usepackage{verbatim} %for the comment block
\usepackage{xcolor,graphicx}
\usepackage{xspace}
\usepackage{algorithm}
\usepackage{algorithmic}
\usepackage{import}
\ifprocs
\hypersetup{colorlinks,allcolors=blue}
\else
\usepackage{ifpdf}
\ifpdf    % we are running pdflatex
\usepackage[pdftex,colorlinks,linkcolor=blue,citecolor=blue,filecolor=blue,urlcolor=blue]{hyperref}
\else    % we are running latex
\usepackage[hypertex,colorlinks,linkcolor=blue,citecolor=blue,filecolor=blue,urlcolor=blue]{hyperref}
\fi
\fi

\ifprocs
\renewcommand{\paragraph}{\subparagraph}
\else
\newtheorem{theorem}{Theorem}[section]
\newtheorem{lemma}[theorem]{Lemma}

\fi

\theoremstyle{plain}

\newcommand{\tuple}[1]{{\langle{#1}\rangle}}

\newcommand{\set}[1]{\left\{ #1 \right\}}

\newcommand\abs[1]{\vert {#1}\vert}

\DeclareMathOperator{\poly}{poly}

\newcommand{\OPT}{\operatorname{OPT}{}}
\newcommand{\opt}{\operatorname{opt}{}}

\DeclareMathOperator{\A}{Arora}
\newcommand{\vecangle}[1]{\operatorname{angle}(#1)}

\newcommand{\R}{\mathbb R}

\newcommand{\B}{\{ 0,1 \}}

\newcommand{\E}{\mathcal E}

\renewcommand{\B}{\mathcal B}
\newcommand{\W}{\mathcal W}

\newcommand{\bpi}{\pmb \pi}

\newcommand{\eps}{\epsilon}

\newcommand{\w}{\bar w}
\newcommand{\EW}{\mathcal E^{\text{win}}}

\newcommand\norm[1]{\left\lVert#1\right\rVert}

\providecommand{\keywords}[1]
{
  \small	
  \textbf{Keywords.} #1
}

\title{Faster Algorithms for Orienteering and $k$-TSP}
\date{}
\ifprocs
\author{Lee-Ad Gottlieb}{Ariel University, Israel}{leead@ariel.ac.il}{}{}
\author{Robert Krauthgamer}{Weizmann Institute of Science, Israel}{robert.krauthgamer@weizmann.ac.il}{}{}
\author{Havana Rika}{Weizmann Institute of Science, Israel}{havana.rika@weizmann.ac.il}{}{}
\authorrunning{L. Gottlieb and R. Krauthgamer and H. Rika} %mandatory.
\Copyright{L. Gottlieb and R. Krauthgamer and H. Rika}%mandatory,

% mandatory:
% mandatory: Please provide 1-5 keywords
\else
\author{Lee-Ad Gottlieb%
\thanks{Ariel University, Israel. Work supported by the Israel Science Foundation grant \#1602/19.
Email: \texttt{leead@ariel.ac.il}
}
\\
Ariel University
\and
Robert Krauthgamer%
\thanks{Weizmann Institute of Science.
  Work partially supported by ONR Award N00014-18-1-2364, the Israel Science Foundation grant \#1086/18, and a Minerva Foundation grant.
  Part of this work was done while the author was visiting the Simons Institute for the Theory of Computing.
  Email: \texttt{robert.krauthgamer@weizmann.ac.il}
}
\\
Weizmann Institute
\and
Havana Rika%
\thanks{Weizmann Institute of Science, Israel.
Email: \texttt{havana.rika@weizmann.ac.il}
}
\\
Weizmann Institute
}

\fi

\begin{document}
\maketitle

\begin{abstract}

We consider the rooted orienteering problem in Euclidean space:
Given $n$ points $P$ in $\R^d$, a root point $s\in P$ and a budget $\B>0$,
find a path that starts from $s$, has total length at most $\B$,
and visits as many points of $P$ as possible.
This problem is known to be NP-hard,
hence we study $(1-\delta)$-approximation algorithms.
The previous Polynomial-Time Approximation Scheme (PTAS) for this problem,
due to Chen and Har-Peled (2008), runs in time
$n^{O(d\sqrt{d}/\delta)} 2^{(d/\delta)^{O(d)}}$,
and improving on this time bound was left as an open problem.
Our main contribution is a PTAS with a significantly improved time complexity of
$n^{O(1/\delta)} 2^{(d/\delta)^{O(d)}}$.

A known technique for approximating the orienteering problem is to reduce it
to solving $1/\delta$ correlated instances of rooted $k$-TSP
(a $k$-TSP tour is one that visits at least $k$ points).
However, the $k$-TSP tours in this reduction must achieve a certain excess guarantee (namely, their length can surpass the optimum length only
in proportion to a parameter of the optimum called excess)
that is stronger than the usual $(1+\delta)$-approximation.
Our main technical contribution is to improve the running time of these
$k$-TSP variants, particularly in its dependence on the dimension $d$.
Indeed, our running time is polynomial even for a moderately large dimension,
roughly up to $d=O(\log\log n)$ instead of $d=O(1)$.

\end{abstract}


\section{Introduction}
The Traveling Salesman Problem (TSP) is of fundamental importance
to combinatorial optimization, computer science and operations research.
It is a prototypical problem for planning routes in almost any context,
from logistics to manufacturing, and is therefore studied extensively.
In this problem, the input is a list of cities (aka sites)
and their pairwise distances,
and the goal is to find a (closed) tour of minimum length
that visits all the sites.
This problem is known to be NP-hard even in the Euclidean case~\cite{GGJ76,Papadimitriou77,Trevisan00},
which is the focus of our work.

One important variant of TSP are \emph{orienteering} problems,
which ask to maximize the number of sites visited
when the (closed or open) tour length is constrained by a given budget.
These problems model scenarios where the ``salesman''
has limited resources, such as gasoline, time or battery-life.
This genre is related to prize-collecting traveling salesman problems,
introduced by Balas~\cite{Balas95},
where the sites are also associated with non-negative ``prize'' values,
and the goal is to visit a subset of the sites while minimizing the total distance traveled and maximizing the total amount of prize collected.
Note that there is a trade-off between the cost of a tour and how much prize it spans.
Another related family is the vehicle routing problem (VRP)~\cite{TV02},
where the goal is to find optimal routes for multiple vehicles visiting a set of sites. These problems arise from real-world applications such as delivering goods to locations or assigning technicians to maintenance jobs.

We consider the \emph{rooted} orienteering problem in Euclidean space,
in which the input is a set of $n$ points $P$ in $\R^d$,
a starting point $s$ and a budget $\B>0$,
and the goal is to find a path that starts at $s$ and visits as many points
of $P$ as possible, such that the path length is at most $\B$.
A $(1-\delta)$-approximate solution is a path satisfying these constraints
(start at $s$ and have length at most $\B$)
that visits at least $(1-\delta)k_{\opt}$ points,
where $k_{\opt}$ denotes the maximum possible,
i.e., the number of points visited by an optimal path.

Arkin, Mitchell, and Narasimhan~\cite{AMN98} designed the first
approximation algorithms for the rooted orienteering problem.
They considered this problem for points in the Euclidean plane
when the desired ``tour'' (network in their context) is a path, a cycle, or a tree,
and achieved a $O(1)$--approximation for these problems.
Blum et al.~\cite{BCKLMM07} and Bansal et al.~\cite{BBCM04}
designed an $O(1)$-approximation algorithm
for rooted path orienteering when the points lie in a general metric space.

%Chen and Har-Peled~\cite{CH08} were the first
%to design a Polynomial-Time Approximation Scheme (PTAS),
%i.e., a $(1-\delta)$-approximation algorithm for every fixed $\delta>0$,
%when the points lie in Euclidean space of fixed dimension.
%Their algorithm reduces the orienteering problem into
%(a certain version of) rooted $k$-TSP,
%and thus the heart of their algorithm is a PTAS for the latter,
%where the approximation is actually with respect to a parameter called \emph{excess}, which can be much smaller than the optimal tour length.
%This parameter was introduced by Blum et al.~\cite{BCKLMM07},
%and was improved and further refined by Chen and Har-Peled.
%The algorithm of Chen and Har-Peled cleverly combines
%two very different divide-and-conquer methods,
%of Arora~\cite{Arora98} and of Mitchell~\cite{Mitchell99}.
%As they point out, a key difficulty in this problem
%is the relative lack of algorithmic tools to handle rigid budget constraints.

Chen and Har-Peled~\cite{CH08} were the first
to design a Polynomial-Time Approximation Scheme (PTAS),
i.e., a $(1-\delta)$-approximation algorithm for every fixed $\delta>0$,
when the points lie in a Euclidean space of fixed dimension.
Their algorithm reduces the orienteering problem into
(a multi-path version of) rooted $k$-TSP,
and thus the heart of their algorithm is a PTAS for the latter,
where the approximation is actually with respect to a parameter called \emph{excess}, which can be much smaller than the optimal tour length.
This follows an earlier approach of Blum et al.~\cite{BCKLMM07},
who introduced the concept of excess-based approximation,
and designed a reduction to a simpler (single-path version of) $k$-TSP.
However, that earlier reduction increases the approximation ratio
by a constant factor and cannot yield a PTAS.
Chen and Har-Peled~\cite{CH08} presented a different reduction,
to a more complicated (multi-path) version of rooted $k$-TSP,
and for the latter problem they designed an algorithm that
cleverly combines two very different divide-and-conquer methods,
of Arora~\cite{Arora98} and of Mitchell~\cite{Mitchell99}.
As they point out, a key difficulty in this problem
is the relative lack of algorithmic tools to handle rigid budget constraints.

We design a PTAS for the rooted orienteering problem that has a better running time than
the known running time
$n^{O(d\sqrt{d}/\delta)}(\sqrt{d} \log n / \delta)^{(d/\delta)^{O(d)}}
= n^{O(d\sqrt{d}/\delta)}(\log n)^{(d/\delta)^{O(d)}}
= n^{O(d\sqrt{d}/\delta)} 2^{(d/\delta)^{O(d)}}$
of Chen and Har-Peled~\cite{CH08}.\footnote{The final equality follows from the fact
that for all $n,x \ge 2$ we have
$(\log n)^x
= 2^{(\log \log n) x}
< 2^{(\log \log n)^2 + x^2}
< n 2^{x^2}$.}
For fixed $\delta$ and small dimension $d$,
the leading term in their running time is about $n^{O(d\sqrt{d}/\delta)}$,
which we improve to $n^{O(1/\delta)}$.
Thanks to this improvement, our algorithm is polynomial even for a moderately large dimension,
roughly up to $d=O(\log\log n)$ instead of $d=O(1)$.\footnote{We note that both our results
and those of \cite{CH08} hold also for the orienteering variant where the input is a pair of
endpoints $s,t$ (instead of just a root point $s$), and the solution is a path connecting
$s$ and $t$.}

\subsection{Our Results}

Our main result is a PTAS for the rooted orienteering problem,
with improved running time compared to that of Chen and Har-Peled~\cite{CH08}.

\begin{theorem}\label{THM: PTAS for Orienteering}
Given as input a set $P$ of $n$ points in $\R^d$, a starting point $s$,
a budget $\B>0$, and an accuracy parameter $\delta\in(0,1)$,
one can compute in time
%$n^{O(1/\delta)}(\log n)^{(d/\delta)^{O(d)}}$,
$n^{O(1/\delta)} 2^{(d/\delta)^{O(d)}}$,
a path that starts at $s$, has length at most $\B$,
and visits at least $(1-\delta)k_{\opt}$ points of $P$,
where $k_{\opt}$ is the maximum possible number of points
that can be visited under these constraints.
\end{theorem}

Similarly to Chen and Har-Peled~\cite{CH08}, our algorithm reduces
the rooted orienteering problem to (a multi-path version of) rooted $k$-TSP,
and the main challenge is to solve the latter problem
with good approximation with respect to the excess parameter.
Their algorithm for $k$-TSP uses Mitchell's divide-and-conquer method
\cite{Mitchell99} based on splitting the space into windows.
These windows contain subpaths of the $k$-TSP path,
and the algorithm finds such subpaths in every window
and then combines them into the requested path.
The leading term $n^{\poly(d)}$ in the running time of Chen and Har-Peled~\cite{CH08}
arises from defining each window via $2d$ independent hyperplanes,
which gives rise to $n^{O(d)}$ possible windows.
%which gives rise to many possible windows.
We define and order the windows in a way that is similar to, but different from, Blum et al.~\cite{BCKLMM07},
%and more effective way using only two points,
which yields at most $n^2$ different windows (see Section~\ref{sec:mk-tsp} for more details).
This improvement is readily seen
in our first technical result (Theorem~\ref{THM: PTAS for k-TSP}),
which provides a PTAS for (a simple version of) rooted $k$-TSP.
For fixed $\delta$ and small dimension $d$,
the leading term in our running time is $n^{O(1)}$.
For our main result, we need to solve a multi-path version
that we call rooted $(m,k)$-TSP.
This problem asks to find $m$ paths that visit $k$ points in total,
when the input prescribes the endpoints of all these $m$ paths
(see Theorem~\ref{THM: PTAS for (m,k)-TSP}). We note that although our result for $k$-TSP can be obtained also by applying the techniques of Blum et al.~\cite{BCKLMM07}, our version of the windows is essential for solving the rooted $(m,k)$-TSP,
and is thus needed to obtain a PTAS for orienteering.

\subsection{Related Work}

The orienteering problem was first introduced by Golden et al.~\cite{GLV87}, and intensely studied since then.
The problem has numerous variants.
For example, Chekuri et al.~\cite{CKP12} designed $(2+\eps)$-approximation algorithm for orienteering in undirected graphs,
and an $O(\log^2 \OPT)$-approximation algorithm in directed graphs.
Gupta et al.~\cite{GKNR15} designed an $O(1)$-approximation algorithm for the best non-adaptive policy for stochastic orienteering. Friggstad et al.~\cite{FS17} introduced the first polynomial-size LP-relaxations for the orienteering problem and its rooted version, and obtained  $O(1)$-approximation algorithms via LP-rounding.
For algorithms for orienteering with deadlines and time-windows see~\cite{BBCM04,CK04,CP05}. A survey on orienteering can be found in~\cite{gunawan16}, and a survey on the vehicle routing problem with profits can be found in~\cite{archetti14}.

\section{Preliminaries}
\label{sec:prelims}

\paragraph{Notation.}
Let $\pi=\langle p_1,\ldots,p_k \rangle$ be a path that visits $k$ points of $P$ in $\R^d$, starting at $p_1$ and ending at $p_k$.
The \emph{length} of $\pi$ is denoted by $\norm{\pi}:=\sum_{j=1}^{k-1}\norm{p_{j+1}-p_j}$,
and let $P(\pi)$ be all the points in $P$ that are visited by $\pi$.
Define the {\em excess} of $\pi$ to be
$$\E(\pi):= \| \pi \| - \| p_k - p_1 \|.$$
Note that the excess of $\pi$ may be considerably smaller than the length of $\pi$.
Similarly, given a set $\Pi$ of $m$ paths, such that each path $\pi_i$, $i\in[m]$, connects endpoints $s_i,t_i$,
we denote the total length of its $m$ paths by
$\norm{\Pi}:=\sum_{i=1}^m \norm{\pi_i}$.
Let $P(\Pi)$ be all points visited by $\Pi$, i.e. $P(\Pi)=\cup_{i=1}^m P(\pi_i)$,
and let the excess of $\Pi$ be
$\E(\Pi):=\sum_{i=1}^m (\norm{\pi_i}-\norm{t_i-s_i})$.

Given a set $P$ of $n$ points and $m$ pairs $s_i,t_i\in P$,
the \emph{rooted $(m,k)$-TSP problem} is to find a set of $m$ paths
$\Pi=\{\pi_i|\ i\in[m]\}$ with minimum total length,
such that each path $\pi_i$ connects endpoints $s_i,t_i$, and $|P(\Pi)|=k$.
A \emph{$\delta$-excess-approximation} to the rooted $(m,k)$-TSP problem is a set of $m$ paths $\Pi=\{\pi_i|\ i\in[m]\}$,
such that each path $\pi_i$ connects endpoints $s_i,t_i$, $|P(\Pi)|=k$,
and $\norm{\Pi} \leq \norm{\Pi^*}+\delta\cdot \E(\Pi^*)$,
where $\Pi^*$ is a solution of minimum length. We define the
\emph{rooted $k$-TSP problem} to be the rooted $(1,k)$-TSP problem.

Given a set $P$ of $n$ points, a budget $\B$, and a starting point $s$,
the \emph{rooted orienteering problem} is the problem of finding a path $\pi^*$
rooted at $s$ which visits the maximum number of points of $P$
under the constraint that $\norm{\pi^*}\leq \B$. Let
$k_{\opt}$ denote the number of points visited by $\pi^*$.
A $(1-\delta)$-approximation to the rooted orienteering problem is a path $\pi$ rooted at $s$
which visits at least $(1-\delta)k_{\opt}$ vertices
under the constraint that $\norm{\pi}\leq \B$.

\paragraph{Algorithms.}
Arora~\cite{Arora98} gave a PTAS for Euclidean TSP which runs in time
$n(\log n)^{(d/\delta)^{O(d)}}$.
He also showed how to modify the algorithm to solve $k$-TSP in time
$k^2 n(\log k)^{(d/\delta)^{O(d)}}$.
This is done by modifying the dynamic program, so that for every
candidate cell the program computes an optimal tour visiting at least
$k'$ points, for each $k' \in [k]$.
(We also note the dependence on $\log k$ instead of $\log n$.)
Another simple modification to Arora's algorithm is to compute
$m$ tours which together visit all points, and this increases the runtime to
$n(2^m\log n)^{(d/\delta)^{O(d)}}$.\footnote{For those familiar with
Arora's construction, we must add to each active portal a list of tours incident
upon it. The factor $2^m$ represents the ensuing increase in the number
of possible configurations.}
Combining these two separate extensions gives a solution to
$(m,k)$-TSP in time
$k^2 n(2^m\log n)^{(d/\delta)^{O(d)}}$.

\section{A $\delta$-excess-approximation algorithm for rooted $(m,k)$-TSP}\label{sec:mk-tsp}

In this section, we present the $\delta$-excess-approximation algorithm for rooted $(m,k)$-TSP.
Later in Section \ref{sec:orient}, we use this algorithm as a subroutine to approximate the orienteering problem.
For purposes of exposition, we will first show how to solve the case $m=1$,
i.e., rooted $k$-TSP, using a plane sweep algorithm (PSA),
and then extend this plane sweep algorithm to general $m$,
i.e., solve $(m,k)$-TSP.

Our algorithm combines ideas from Blum et al.~\cite{BCKLMM07} and Chen and Har-Peled~\cite{CH08}.
In Section~\ref{sec:sub-ktsp} we solve $k$-TSP using the techniques of Blum et al.~\cite{BCKLMM07}, but with a critical modification that uses the Euclidean (rather than metric) setting and allows extension of our algorithm to the more general $(m,k)$-TSP in Section~\ref{sec:sub-mktsp}.

\subsection{Algorithm for rooted $k$-TSP}\label{sec:sub-ktsp}

We present $\delta$-excess-approximation algorithm for rooted $k$-TSP, as follows.

\begin{theorem}\label{THM: PTAS for k-TSP}
  Given as input the endpoints $s,t \in \R^d$, a set of $n$ points $P\subset \R^{d}$, an integer $2\leq k\leq n$, and an accuracy parameter $\delta\in(0,1)$,
  there is an algorithm that runs in time
%$n^{O(1)}(\log k)^{(d/\delta)^{O(d)}}$
$n^{O(1)} 2^{(d/\delta)^{O(d)}}$
  and finds a $k$-TSP path from $s$ to $t$ of length at most $\OPT + \delta\cdot \E$,
  where $\OPT$ is the minimum length of a $k$-TSP path from $s$ to $t$,
  and its excess is denoted by $\E=\OPT -\norm{t-s}$.
\end{theorem}

The rest of this section is devoted to the proof of Theorem \ref{THM: PTAS for k-TSP}.
Before introducing the construction and proof, let us present the intuition behind it.
First let us rotate the space so that $s,t$ both lie on the $x$-axis,
with the $x$-coordinate of $s$ smaller than the $x$-coordinate of $t$.
Now suppose that the optimal path $\pi^*$ is monotonically increasing in $x$.
In this case, the optimal tour could be computed in quadratic time by a simple PSA,
a dynamic programming algorithm defined by a plane orthogonal to the $x$-axis and sweeping
across it from $x=-\infty$ to $x=\infty$.
For every encountered point $p \in P$, the algorithm must determine the optimal tour from
$s$ to $p$ visiting $k'$ points for all $k' \in [k]$, and since all edges are $x$-monotone,
all these $k'$ points must have been encountered previously by the sweep.
It follows that in this case the optimal $k'$-TSP tour ending at $p$ can be computed
by taking an optimal $(k'-1)$-TSP tour ending at each previously encountered point
$p'$, extending it to $p$ by one edge of length $\norm{p-p'}$, and then choosing among these tours (all choices of $p'$) the one of minimum cost.
This PSA ultimately computes the best $k$-TSP ending at each possible point, and the minimum
among these is the optimal tour.

The difficulty with the above approach is that the optimal tour might be non-monotone in the $x$-coordinate. It may contain {\em backward} edges,
while the PSA algorithm described above can only handle {\em forward} edges.
However, for an edge facing backwards, its entire length accounts for excess
in the tour. Hence, in the space (or more precisely, window) containing the
backward edge, we can afford to run Arora's $k$-TSP algorithm,
and pay $(1+\delta)$ times the entire tour cost in the window.
This motivates an algorithm which combines a sweep with Arora's $k$-TSP algorithm.
We proceed with the actual proof,
denoting the $x$-coordinate of a point $p$ by $p[0]$.

This algorithm is similar to that of Blum et al.~\cite{BCKLMM07}, which defines every window using only 2 points, and distinguishes between windows with only forward edges and windows that contain backward edges (called type-1 and type-2 in~\cite{BCKLMM07}). They use the MCP routine of~\cite{CGRT03} to approximate the minimum length of a path inside windows with backward edges, and stitch all the subpaths together by dynamic programming. There are two main differences in our algorithm. First, we run Arora's algorithm on the windows with backward edges. Second, our dynamic programming goes over the points in a different order. They order the points in increasing order of distance from $s$ (the starting point), while we order the points by their $x$-coordinate.

\begin{proof}[Proof of Theorem~\ref{THM: PTAS for k-TSP}]
We rotate the space so that $s,t$ lie on the $x$-axis, and then order all points based on increasing $x$-coordinate.
For two distinct points $p,q\in P$,
we say that \emph{$p$ is before $q$}, denoted $p<q$,
if $p$'s $x$-coordinate $p[0]$ is smaller than $q$'s $x$-coordinate $q[0]$;
otherwise, we say that \emph{$q$ is after $p$}, denoted $q>p$.
We can make an infinitesimally small perturbation on the points
to ensure that $p[0]\neq q[0]$ for all distinct points $p,q\in P$.

We define a {\em window} in $\R^d$ to be the space between (and including) two $(d-1)$-dimensional hyperplanes orthogonal to the $x$-axis.
For every point pair $a,b\in P$ with $a \le b$ (which means we allow $a=b$),
let $\w_{a,b}$ be a window of width $b[0]-a[0]$
containing $a,b$ on its respective ends. Note that a window is bounded in the $x$ direction and
unbounded in all other directions, and also that a window may have width $0$  (if $a=b$) and then it can contain at most one point of $P$.
We denote the points of $P$ contained in a window $\w_{a,b}$
by $P(\w_{a,b}):=\{p\in P \mid a\leq p\leq b\}$.
Let $\W:=\{\w_{a,b} \mid a,b\in P, a\leq b\}$,
and so $|\W|\leq O(n^2)$.

\paragraph{Algorithm.}
For some $\delta' = \Theta(\delta)$ to be specified below,
let $\A(a,b,c,d,k)$ be the output of Arora's
$(1-\delta')$-approximate $k$-TSP algorithm on the set $P(\w_{a,b})$ and
tour endpoints $c,d$ with $a \le c,d \le b$;
recall this algorithm returns the length of a near-optimal tour.
For every $a,b,c,d\in P$ and $k'\in [k]$, we precompute $\A(a,b,c,d,k')$,
see Figure~\ref{FGR: abcd window}.
We then order the points in $P$ by their $x$-coordinate as $p_1,\ldots,p_n$, and let $P_i:=\{p\in P \mid p\leq p_i\}$.
The algorithm sweeps over $p_1,\ldots,p_n$ (i.e., by their $x$-coordinate),
and upon encountering point $p_i$, it calculates for every $k'\in [k]$
a path from $s$ to $p_i$ visiting $k'$ points of $P_i$.
However, the algorithm first precomputes approximate subpaths
on many windows using Arora's algorithm,
and thus the sweep is actually stitching these subpaths together into a global solution.

\begin{figure}
  \centering
  \includegraphics[angle=0,width=0.6\textwidth]{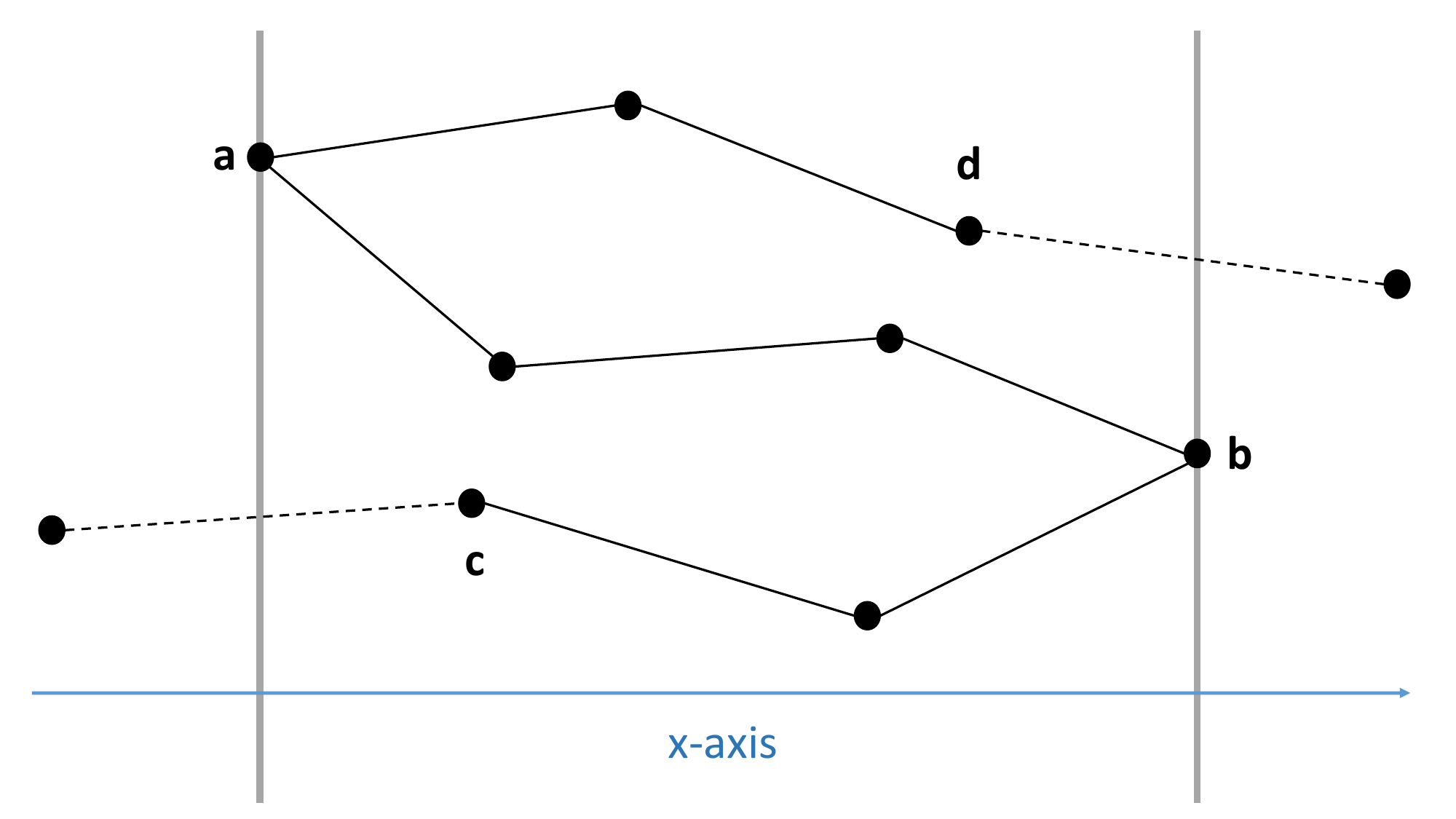}
  \caption{The window $\w_{a,b}$ is the space between the two thick gray lines,
    The solid black lines represent the subpath $\A(a,b,c,d,8)$,
    and the two dashed lines represent extending it outside that window.
  }
  \label{FGR: abcd window}
  \mbox{}\hrule
\end{figure}

Let $V$ be a $3$-dimensional dynamic programming table
with each entry $V(p_i, d,k')$ for $p_i\in P$, $d\in P_i$ and $k' \in [k]$,
containing the length of an already computed path from $s$ to $d$ that visits $k'$ points in $P_i$.
To initialize the table,
for all $d,p_i$ satisfying $d \le p_i= s$ and $k' \in [k]$,
we fix entries
$V(p_i, d,k')=\A(p_1,p_i,s,d,k')$.
All other entries are set to $\infty$.
(Note that this forces all paths to begin at $s$, even if portions of
those paths travel to the left of $s$.)
The algorithm then considers each $p_i > s$ in increasing order,
and calculates the entries for all $d$ satisfying $s < d \le p_i$ and all $k' \in [k]$,
by choosing the shortest path among several possibilities, as follows.

\begin{align*}
  V(p_i, d,k')
  = \min\Big\{
     V(&p_j, d', k'') + \norm{d'-c} + \A(p_{j+1},p_i,c,d,k'-k'') \mid \\
     & p_j\in P_i,\ d'\in P_j,\ c \in P_i\setminus P_j,\ k''<k'
    \Big\}
\end{align*}

The path associated with this entry
combines a previously computed path from $s$ to $d'$ visiting $k''<k'$ points in $P_j$
with an Arora subpath connecting endpoints $c,d$ inside window $\w_{p_{j+1},p_i}$
and visiting $k'-k''$ points in that window.
Connecting these two paths using an edge $(d',c)$ produces
a path from $s$ to $d$ that visits $k'$ points in $P_i$.
After populating the table, the algorithm reports the entry $V(p_n,t,k)$.

The above algorithm computes the {\em length} of a path,
but as usual it extends easily to return also the path itself.
It remains to prove that the returned path has length at most
$\OPT + \delta\cdot\E$.

\paragraph{Correctness.}
We will show that there exists a solution of length at most $\OPT + \delta\cdot\E$ that is considered by the dynamic program.
Let $\pi^*$ be an optimal path from $s$ to $t$ visiting $k$ points of $P$, i.e., $|\pi^*|=\OPT$.
Given a path $\pi$ and two points $p,q\in P(\pi)$,
denote by $\pi(p,q)$ the subpath of $\pi$ from $p$ to $q$.

The solution produced by our PSA represents a set of windows connected by edges between them.
However, some of these windows may be trivial and contain only a single point (and no edges),
so in fact the PSA produces a solution which is a set of non-trivial windows connected by
$x$-monotone subpaths (i.e., subpaths with only forward edges).
As such, our analysis will similarly split $\pi^*$ into windows and subpaths,
where the windows contain all the backward edges of $\pi^*$,
and the remaining edges constitute $x$-monotone subpaths.
Assume that there are $\ell$ maximal backward subpaths in $\pi^*$
(meaning that all edges of these subpaths face backwards),
denoted $\pi^*(b_i,a_i)$ for $i\in[\ell]$, where $a_i<b_i$.
Clearly, $\pi^*(b_i,a_i)$ is fully contained in window $\w_{a_i,b_i}$.
Since these windows may overlap (have non-empty intersection),
we repeatedly merge overlapping windows,
i.e., replace any two overlapping windows $\w_{a_i,b_i}, \w_{a_j,b_j}$
by the united window $\w_{\min \{a_i,a_j \}, \max \{b_i,b_j \} }$,
until no overlapping windows remain.
We thus assume henceforth that the $l$ windows $\w_{a_i,b_i}$
are pairwise disjoint and denote $\W^*:=\{\w_{a_i,b_i} \mid i\in[\ell]\}$.
See Figure~\ref{FGR: st segment} for illustration.

\begin{figure}
  \centering
  \includegraphics[angle=0,width=0.8\textwidth]{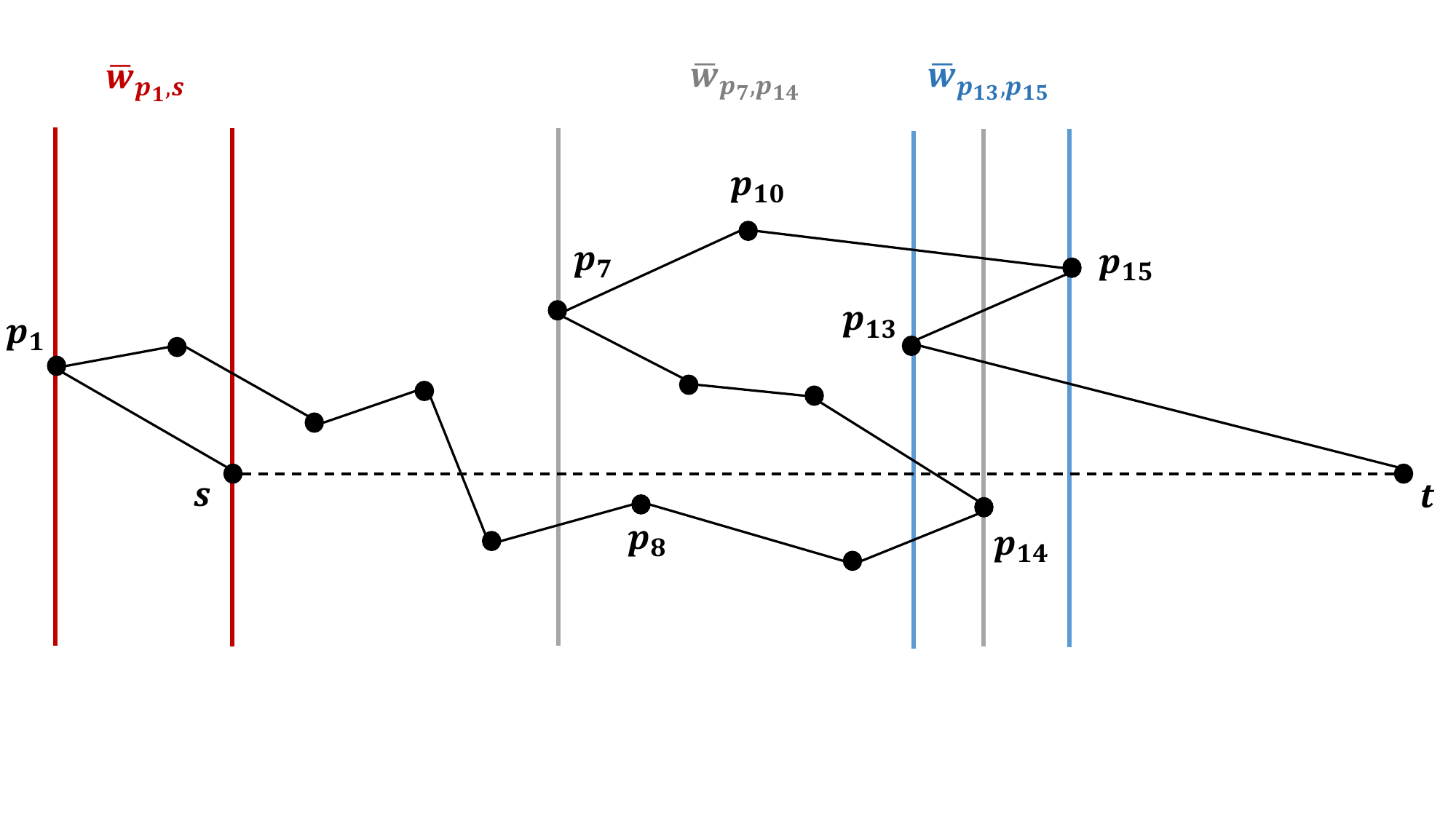}
  \caption{The solid black lines represent an optimal $k$-TSP path $\pi^*$ between $s$ and $t$ for $k=16$.
    It has 3 backward-facing subpaths
    $\pi^*(s,p_1)$, $\pi^*(p_{14},p_7)$ and $\pi^*(p_{13},p_{15})$.
    After merging overlapping windows $\w_{p_7,p_{14}}$ and $\w_{p_{15},p_{13}}$,
    we have $W^*=\{\w_{p_1,s},\w_{p_7,p_{15}}\}$.
    The window $\w_{p_7,p_{15}}$ contains the subpath $\pi^*(p_8,p_{13})$,
    and thus $\EW(\w_{p_7,p_{15}},p_8,p_{13}) = \norm{\pi^*(p_8,p_{13})}-|p_{13}[0]-p_8[0]|$.
  }
  \label{FGR: st segment}
  \mbox{}\hrule
\end{figure}

Having merged overlapping windows, we have an ordered set of windows
where every two successive windows are connected by an $x$-monotone
subpath of $\pi^*$.
Now for a window $\w_{a,b}\in \W^*$, let $c^*(\w_{a,b})$ and $d^*(\w_{a,b})$
be the entry and exit points of the optimal path $\pi^*$ inside $\w_{a,b}$;
notice these points are necessarily unique.

Recall that the excess of $\pi^*$ is defined as $\E(\pi^*)=\norm{\pi^*}-\norm{t-s}$.
Let $E_f$ be all the edges in $\pi^*$ that face forwards, and denote by $\norm{E_f}$ their total length.
Similarly, let $E_b$ be all edges in $\pi^*$ that face backwards, and denote by $\norm{E_b}$ their total length.
Clearly $\norm{E_f} \geq \norm{t-s}$, hence every edge that faces backwards contributes its entire length to the
excess $\E(\pi^*)$, i.e.,
\begin{equation} \label{eq:Eb}
  \E(\pi^*)
  = \norm{\pi^*}-\norm{t-s}=\norm{E_f}+\norm{E_b}-\norm{t-s}
  \ge \norm{E_b}.
\end{equation}
We now define the excess of a window $\w_{a,b}\in \W^*$
with endpoints $c^*=c^*(\w_{a,b})$ and $d^*=d^*(\w_{a,b})$
to be
\[
  \EW(\w_{a,b}) := \norm{\pi^*(c^*,d^*)}-|d^*[0] - c^*[0]|,
\]
which is non-negative because
$\norm{\pi^*(c^*,d^*)}
  \ge | b[0]-a[0] |
  \ge | d^*[0] - c^*[0] |$.
Because the windows in $\W^*$ are pairwise disjoint, it is immediate that
\begin{equation} \label{eq:EW}
  \sum_{\w_{a,b}\in \W^*}\EW(\w_{a,b})
  \leq \E(\pi^*).
\end{equation}
See Figure~\ref{FGR: st segment} for illustration.

Applying Arora's $k$-TSP algorithm on window $\w_{a,b} \in \W^*$
with endpoints $c^*=c^*(\w_{a,b})$ and $d^*=d^*(\w_{a,b})$
returns a path of length at most $(1+\delta')\norm{\pi^*(c^*,d^*)}$,
and we would like to bound $\delta' \norm{\pi^*(c^*,d^*)}$
relative to the excess $\E(\pi^*)$.
To this end, we first bound it relative to
backward edges and excess in that subpath/window,
by claiming that
\begin{equation} \label{eq:ExcessPlusBackwards}
  \norm{\pi^*(c^*,d^*)}
  \le 2 \max\big\{ \norm{E_b\cap \pi^*(c^*,d^*)} , \EW(\w_{a,b}) \big\},
\end{equation}
where $E_b\cap \pi^*(c,d)$ denotes the backward-facing edges in $\pi^*(c,d)$.
Indeed, the claim holds trivially if
$\norm{\pi^*(c^*,d^*)} \le 2\norm{E_b\cap \pi^*(c^*,d^*)}$,
and otherwise we have
$\norm{\pi^*(c^*,d^*)}
  > 2\norm{E_b\cap \pi^*(c^*,d^*)}
  \ge 2 | d^*[0] - c^*[0] |$,
and thus
$\EW(\w_{a,b})
  = \norm{\pi^*(c^*,d^*)}-|d^*[0]-c^*[0]|
  \ge \frac12 \norm{\pi^*(c^*,d^*)}$,
as claimed.

It follows that applying Arora's algorithm with $\delta'=\frac{1}{4}\delta$
on each of the non-overlapping windows in $\W^*$
will approximate the optimum $\norm{\pi^*}$ within total additive error
\begin{align*}
  \sum_{\w_{a,b}\in \W^*} & \delta' \cdot \norm{\pi^*(c^*(\w_{a,b}),d^*(\w_{a,b}))}
  \\
  & \le 2\delta' \sum_{\w_{a,b}\in \W^*}
    \big[\EW(\w_{a,b}) + \norm{E_b\cap \pi^*(c^*(\w_{a,b}),d^*(\w_{a,b}))} \big]
  & \text{by~\eqref{eq:ExcessPlusBackwards}}
  \\
  &\le 2\delta' \big[ \E(\pi^*) + \norm{E_b} \big]
    \le 4\delta' \cdot \E(\pi^*)
    = \delta \cdot \E
  & \text{by~\eqref{eq:EW} and~\eqref{eq:Eb}}.
\end{align*}

Our PSA is a dynamic program that optimizes over many combinations of windows,
including the above collection $\W^*$,
and thus the path that it returns, which can be only shorter,
must be a $\delta$-excess-approximation to the optimal $k$-TSP solution $\pi^*$.

\paragraph{Running time.}
The table $V$ has $O(n^3)$ entries, and computing each entry requires
consulting $O(n^3)$ other entries.
In addition, one has to invoke Arora's $k$-TSP algorithm $O(n^3)$ times,
each executed in time $k^2 n(\log k)^{(d/\delta)^{O(d)}}$.
Thus, the total running time is indeed
$n^{O(1)}(\log k)^{(d/\delta)^{O(d)}} =
n^{O(1)} 2^{(d/\delta)^{O(d)}}$.
This completes the proof of Theorem~\ref{THM: PTAS for k-TSP}.
\end{proof}

\subsection{Algorithm for rooted $(m,k)$-TSP}\label{sec:sub-mktsp}

Having shown how to construct a PSA for $k$-TSP, we extend this result to
the more general $(m,k)$-TSP.
We can prove the following theorem, which is the extension of
Theorem \ref{THM: PTAS for k-TSP} to multiple tours:

\begin{theorem}\label{THM: PTAS for (m,k)-TSP}
There is an algorithm that,
given as input $m$ source-sink pairs $s_i,t_i\in\R^d$ for $i \in [m]$,
a set of $n$ points $P \subset \R^d$,
an integer $1 \le k \le n$,
and an accuracy parameter $\delta\in(0,1)$,
runs in time
%$n^{O(m)}(\log n)^{(md/\delta)^{O(d)}}$
$n^{O(m)} 2^{(md/\delta)^{O(d)}}$
and reports $m$ paths, one from each $s_i$ to its corresponding $t_i$,
that together visit $k$ points of $P$
and have total length at most
$\norm{\Pi^*} + \delta \cdot \E(\Pi^*)$,
where $\Pi^*$ is the minimum total length of $m$ such paths.
\end{theorem}

As before, we first provide the construction, and then demonstrate correctness.
The construction closely parallels that of the PSA for $k$-TSP,
being a collection of $x$-monotone multi-paths connecting windows.

For points $s,t\in\R^d$, let $st$ be the directed line segment connecting them.
Let the \emph{angle} of $st$ be the angle of its direction vector
to the $x$-axis.
Given a path $\pi$ with endpoints $s,t$, we define the angle of $\pi$
to be the angle of the vector $st$.

Given a set $\Pi$ of $m$ paths with respective endpoints $s_i,t_i$ for $i \in [m]$
the space may be rotated and (if necessary) some values
$s_i,t_i$ swapped to ensure that in the resulting space
each directed path has angle in the range
$[0,\frac{\bpi}{2} - \frac{1}{m'}],$\footnote{We use $\bpi$ to denote the
mathematical constant, and $\pi$ to denote a path.}
where $m' = 8m^{3/2}$
(see Lemma \ref{lem:direction}).
We execute this rotation step before the run of the multi-path PSA.
The crux is that this direction is ``good enough'' for each of the $m$ paths,
meaning that it is effective as a sweep direction
for all the $m$ paths simultaneously.
In contrast, the ordering of Blum et al.~\cite{BCKLMM07}
(according to distances from a starting point $s$)
must use a single starting point and does not extend to multiple paths.

\ifprocs
\begin{figure}
    \begin{subfigure}[ t ]{0.5\textwidth }
        \centering
        \includegraphics[ width =1\textwidth ]{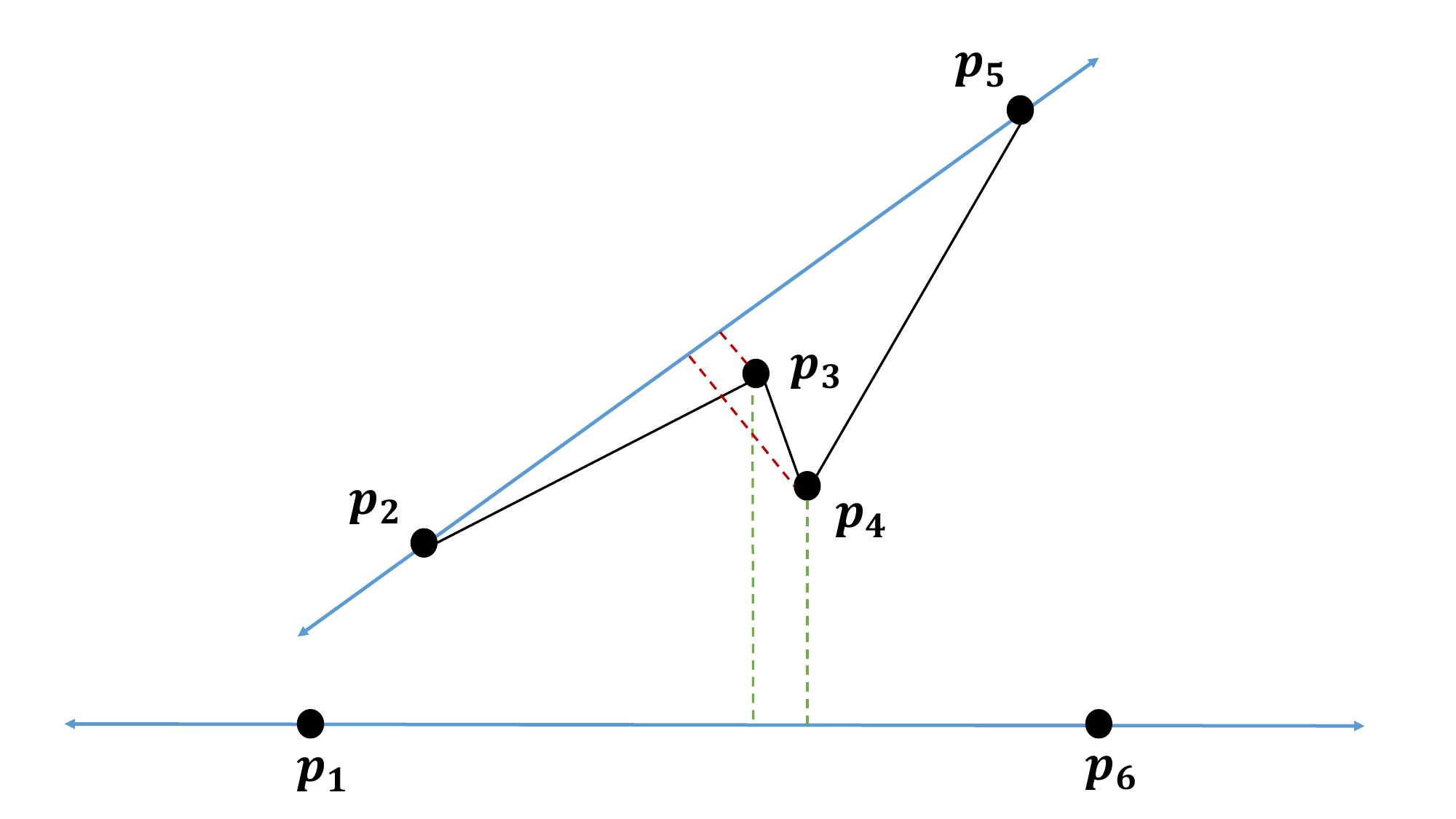}
        \caption{}\label{FGR: mkTSPsegmentA}
    \end{subfigure}\hfill
    \begin{subfigure}[ t ]{0.5\textwidth }
        \centering
        \includegraphics[ width =1\textwidth ]{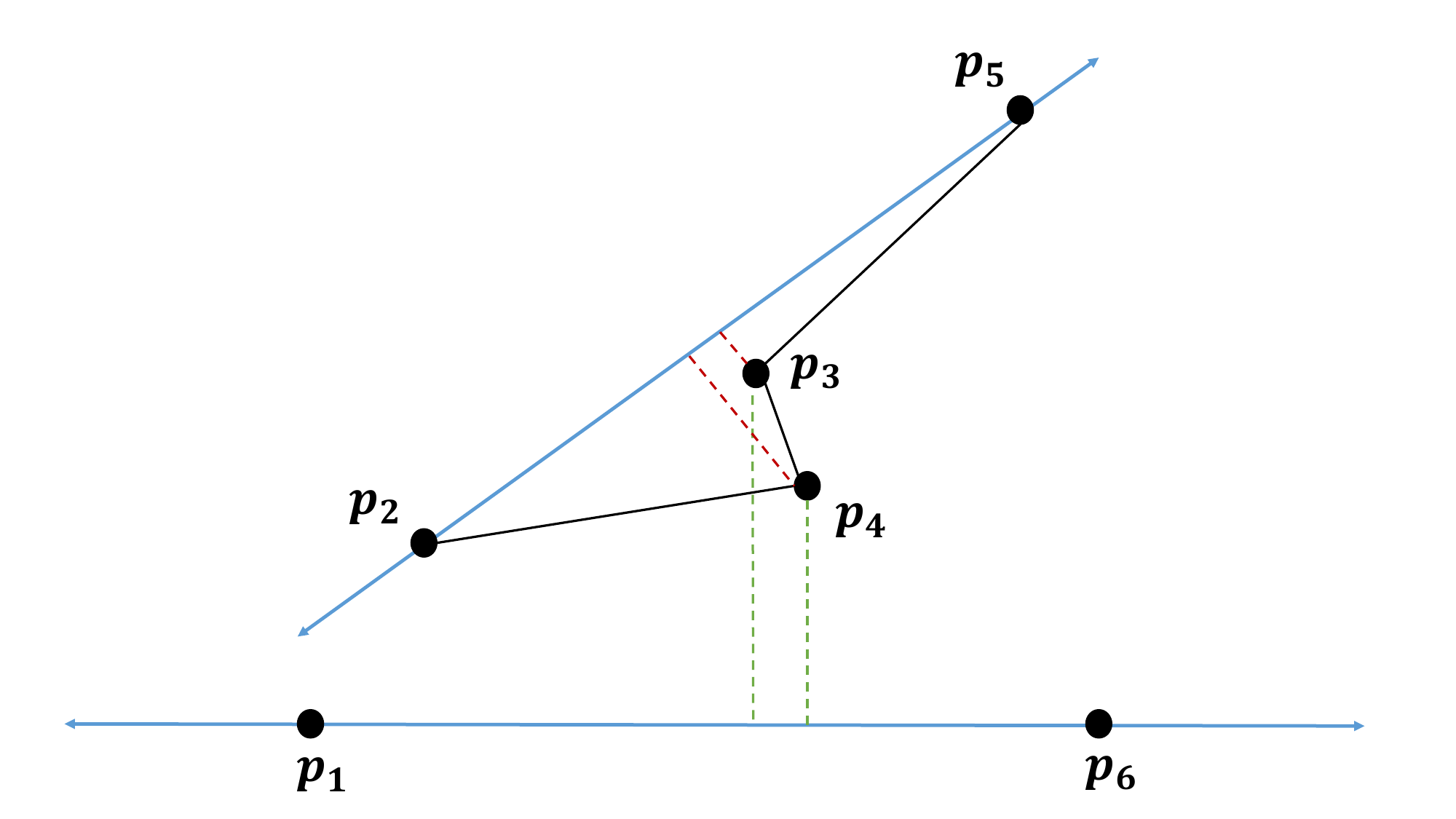}
        \caption{}\label{FGR: mkTSPsegmentB}
    \end{subfigure}
    \caption{ Both figures consist of the same points $p_1,\ldots,p_6 \in R^d$,
and the two directions $h_1$ and $h_2$ are determined by the points $p_1,p_6$ and $p_2,p_5$,
respectively.
Figures~\ref{FGR: mkTSPsegmentA} and~\ref{FGR: mkTSPsegmentB} differ only in the path that connects
points $p_2,p_3,p_4,p_5$. In Figure~\ref{FGR: mkTSPsegmentA} the edge $(p_3,p_4)$ is a backward-facing
edge with respect to $h_2$ but not $h_1$, while in Figure~\ref{FGR: mkTSPsegmentB} the edge $(p_4,p_3)$ is a
backward-facing edge with respect to $h_1$ but not $h_2$.}\label{FGR: mkTSPsegment}
\end{figure}

\else
\begin{figure}
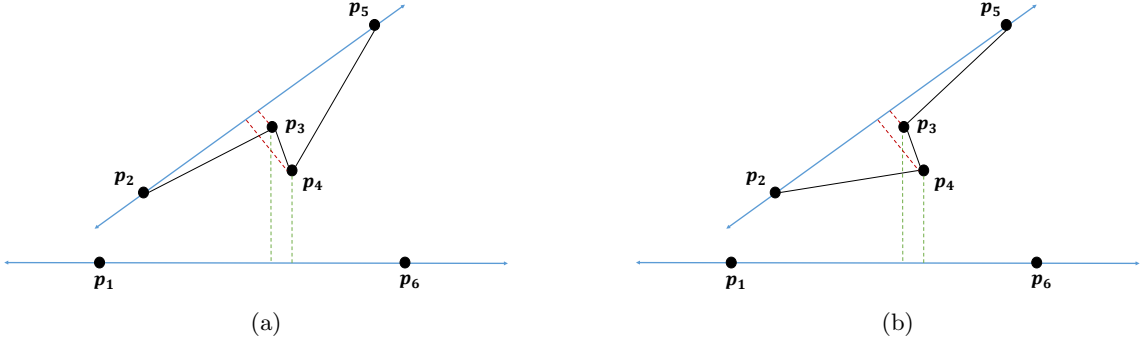

    \centering
    \subfloat[]{{\includegraphics[width=7.5cm]{mkTSPsegmentA}\label{FGR: mkTSPsegmentA} }}
    \qquad
    \subfloat[]{{\includegraphics[width=7.5cm]{mkTSPsegmentB}\label{FGR: mkTSPsegmentB} }}
    \caption{Both figures consist of the same points $p_1,\ldots,p_6 \in R^d$, and the two hyperplanes $h_1$ and $h_2$ are determined by the points $p_1,p_6$ and $p_2,p_5$ correspondingly. Figures~\ref{FGR: mkTSPsegmentA} and~\ref{FGR: mkTSPsegmentB} differ by the path that connects the points $p_2,p_3,p_4,p_5$. In Figure~\ref{FGR: mkTSPsegmentA} the edge $(p_3,p_4)$ is a backward-facing edge with respect to $h_2$, while in Figure~\ref{FGR: mkTSPsegmentB} the edge $(p_4,p_3)$ is a backward-facing edge with respect to $h_1$.}\label{FGR: mkTSPsegment}
\end{figure}
\fi

\paragraph{Construction.}
The construction handles $m$ paths simultaneously.
Let $S,T$ be arrays of length $m$,
with entries $S[j], T[j]$ corresponding to the source-sink pair
of the $j$-th path.
Let a window be defined as in Section \ref{sec:sub-ktsp}.
For some $\delta' = \Omega_m(\delta)$ to be specified below,
let $\A(a,b,S,T,k)$ be the output of Arora's
$(1+\delta')$-approximate
$(m,k)$-TSP algorithm on the set $P(\w_{a,b})$
and tour endpoint arrays $S,T$.
(We may assume for simplicity that
$a \le S[j],T[j] \le b$ for all $j$.
If both $S[j],T[j]$ are null, the algorithm
will ignore the $j$-th path. If exactly one is null, the algorithm
will return $\infty$.)
For every $a,b \in P$, $S,T \subset P^m$ and $k' \in [k]$,
we precompute $\A(a,b,S,T,k')$.
The algorithm then sweeps the $x$-axis from left to right as before
to calculate the solution to subproblems up to a point $p_i \in P$.

Similarly to what was done above,
let $V$ be a $4$-dimensional lookup table
with an entry $V(p_i,S,T,k')$ for every
$p_i \in P$, $S,T \in P_i^m$ and $k' \in [k]$,
that contains the length of computed paths from
sources $S$ to sinks $T$ that together visit $k'$
points in $P_i$.
For the initialization, we add a dummy point $p_0$ to $P$
and initialize the single entry
$V(p_0,S,T,0)$,
where arrays $S,T$ contain all null points, to be 0.
We initialize all other table entries to $\infty$.
Define the distance from a point to a null point
(or between two null points) to be 0.
The algorithm considers each $p_i \in P$ $(i \ge 0)$
in increasing order, and computes the entries for
all $S,T \in P_i^m$ by choosing the shortest path among
several possibilities, as follows:

\begin{align*}
  V(p_i, S, T, k')
  = \min\Big\{
     V(&p_j, S_1, T_1, k'') + \sum_{l=1}^{m} \norm{T_1[l]-S_2[l]} + \A(p_{j+1},p_i,S_2,T_2,k'-k'') \mid \\
     & p_j\in P_i,\ S_1,T_1 \in P_j^m,\ S_2,T_2 \in (P_i\setminus P_j)^m,\ k''<k'
    \Big\}
\end{align*}

As before, this computation combines the length of a previously
computed approximated shortest path to a new Arora multi-path.
However, we add to the above description feasiblity requirements,
which are sufficient to ensure the validity of the final tour.
First note that the invocation to Arora's algorithm on arrays $S,T$ ensures
that a non-infinite solution is possible only if
$S[j],T[j]$ are both null or both non-null for all $j$.
We further require for $S,S_1,S_2$ and all $j$ that
$S_1[j] = S[j]$, unless $S_1[j]$ is null, in which case
we require that $S_2[j] = S[j]$. This ensures that the computed subtour
has source $S[j]$.
Likewise, we require for $T,T_1,T_2$ and all $j$ that
$T_2[j] = T[j]$, unless $T_2[j]$ is null, in which case
we require that $T_1[j] = T[j]$. This ensures that the computed subtour
has sink $T[j]$.
Also, if $S[j]$ ($T[j]$) is null, then
$S_1[j],S_2[j]$ ($T_1[j],T_2[j]$) must be null as well.
This ensures that the subproblems do not feature additional tours.
If these requirements are not met for some set $\{ S,S_1,S_2,T,T_1,T_2 \}$,
then the table value is not changed.
After populating the table, the algorithm reports the entry $V(p_n,S,T,k)$
for $S,T$ containing the sources and sinks of the master problem.

\paragraph{Correctness.}
We must show that there is a set of paths $\Pi$ with
$\norm{\Pi} \le \norm{\Pi^*} + \E(\Pi^*)$
which can be found by the above algorithm.
As before, it suffices to show that the optimal solution can be
divided into windows connected by $x$-monotone paths.

In the analysis of the $k$-TSP algorithm, we used the fact that any
backward-facing edge contributes its entire length to the excess. This
does not hold in the $(m,k)$-TSP case, as the definition of
``backwards'' remains with respect to the $x$-axis, but excess
is measured with respect to the angle of the relevant path,
see Figure~\ref{FGR: mkTSPsegment}.
To address this, we will require the following lemma:

\begin{lemma}\label{lem:projection}
Given parameter $0 \le \gamma \le 1$
(where $\gamma$ is a measure in radians)
and directed path $\pi$ with angle $\phi$ to the $x$-axis and edge-set $E$,
let $E' \subset E$ consist of directed edges with angles in the range
$[\phi-\gamma, \phi+\gamma]$.
Then
we have
\[
   \sum_{e \in E\setminus E'} \| e \|
   \le \frac{24}{11 \gamma^2} \E(\pi).
\]
\end{lemma}

\begin{proof}
For edge $e \in E$, let $p(e)$ be the length of the projection of $e$
onto segment $st$, where $s,t$ are the endpoints of $\pi$.
We can charge each edge $e$ a share of the excess as follows.
Define $\E(e) = \|e\| - p(e)$;
then by the triangle inequality $\norm{s-t} \leq \sum_{e\in E} p(e)$,
and thus $\sum_{e \in E} \E(e) = \norm{\pi} - \sum_{e \in E} p(e)  \le \E(\pi)$.
Now consider an edge $e \in E\setminus E'$.
Recalling the Taylor expansion
$\cos(x) = 1 - \frac{x^2}{2!} + \frac{x^4}{4!} - \ldots$,
we have that
\[
p(e)
\le \| e\| \cos (\gamma)
\le \| e\| \left( 1- \frac{\gamma^2}{2} + \frac{\gamma^4}{24} \right)
\le \| e\| \left( 1-\frac{11}{24}\gamma^2 \right),
\]
and so
$ \E(e)
=\|e\| - p(e)
\ge \frac{11}{24}\gamma^2 \| e\|$.
It follows that
\[
\frac{11 \gamma^2}{24} \sum_{e \in E \setminus E'} \norm{e}
\le \sum_{e \in E \setminus E'} \E(e)
\le \E(\pi).
\]
\end{proof}

Now take optimal tour $\Pi^*$, and let $\W^*$ be defined as in Section \ref{sec:sub-ktsp},
that is consisting of mergers of maximal windows which together
cover all backward-facing edges (where the direction is defined
with respect to the $x$-axis).
The edges not in windows of $\W^*$ constitute forward-facing paths.
Now consider some window $\w_{a,b} \in \W^*$, and we will show that we
can afford to run Arora's $(m,k)$-TSP on this window
with sufficiently small parameter $\delta'$.

We begin with the set $E_b$ of backward-facing edges of $\Pi^*$.
As the angle of all paths in $\Pi^*$ was shown above to be in the range
$[0,\frac{\bpi}{2} - \frac{1}{m'}]$ radians
(and backward-facing edges necessarily have angle greater than $\frac{\bpi}{2}$)
we can apply Lemma \ref{lem:projection} with parameter
$\gamma = \frac{1}{m'}$,
and conclude that
\[
\norm{E_b}
\le \frac{24(m')^2}{11}\E(\Pi^*) = O (m^3) \cdot \E(\Pi^*).
\]

We now turn to the set $E_f$ of forward-facing edges of $\Pi^*$.
Set $E'_f \subset E_f$ will contain edges whose angle is very close
to the angle of their path. More precisely, let path
$\pi^*_i$ have angle $\phi_i$.
$E'_f$ includes every edge of every path $\pi^*_i$
with angle in the range
$\left[ \phi_i - \frac{1}{2m'}, \phi_i + \frac{1}{2m'} \right]$.
Now consider edges of $E_f \setminus E'_f$:
Applying Lemma \ref{lem:projection} with parameter
$\gamma = \frac{1}{2m'}$,
we have that
\[
\norm{E_f \setminus E'_f}
\le \frac{24(2m')^2}{11}\E(\Pi^*)
= O(m^3) \cdot \E(\Pi^*).
\]

Finally, we now turn to the set $E'_f$.
First recall that by the Taylor expansion,
$\sin(x) = x - \frac{x^3}{3!} + \ldots$
Each edge in $E_f'$ accounts for a
progression in the $x$-direction of at least
\[
\norm{e} \cos \left( \frac{\bpi}{2} - \frac{1}{2m'} \right)
= \norm{e} \sin \left( \frac{1}{2m'} \right)
\ge \norm{e} \left( \frac{1}{2m'} - \frac{1}{6(2m')^3} \right)
> \norm{e} \frac{1}{3m'}.
\]
Now consider window $\w_{a,b} \in \W^*$,
and its associated paths
$\pi^*_i(c^*_i(\w_{a,b}), d^*_i(\w_{a,b}))$.
Clearly $\pi^*_i(c^*_i(\w_{a,b}), d^*_i(\w_{a,b}))$
cannot progress in the $x$-direction inside the
window for more than its length without heading backwards,
and so
\[
\frac{1}{3m'} \norm{\pi^*_i(c^*_i(\w_{a,b}), d^*_i(\w_{a,b})) \cap E'_f}
\le |b[0]-a[0]| + \norm{\pi^*_i(c^*_i(\w_{a,b}), d^*_i(\w_{a,b})) \cap E_b}.
\]

Now recall that by construction, each window $\w_{a,b}\in \W^*$ contains
backward-facing edges whose lengths sum to at least the window length, and so
$\sum_{\w_{a,b}\in \W^*} |b[0]-a[0]| \le \norm{E_b}$.
Applying Arora's algorithm with parameter $\delta'$
on each of the non-overlapping windows in $\W^*$
will approximate the optimum $\norm{\Pi^*}$ within total additive error
\begin{align*}
  \sum_{\w_{a,b}\in \W^*} & \sum_{i=1}^m \delta'
    \norm{\pi^*_i(c^*_i(\w_{a,b}),d^*_i(\w_{a,b}))}
  \\
  & \le \delta' \Big(
	\norm{E_b} + \norm{E_f \setminus E'_f}
	+ \sum_{\w_{a,b}\in \W^*} \sum_{i=1}^m \norm{E'_f \cap \pi^*_i(c^*_i(\w_{a,b}), d^*_i(\w_{a,b})) }
	\Big)
  \\
  & \le \delta' \Big(
	\norm{E_b} + \norm{E_f \setminus E'_f}
	+ \sum_{\w_{a,b}\in \W^*} \sum_{i=1}^m
	3m' \Big[ |b[0]-a[0]| + \norm{E_b \cap \pi^*_i(c^*_i(\w_{a,b}), d^*_i(\w_{a,b})) } \Big]
	\Big)
  \\
  & \le \delta' \Big(
	\norm{E_b} + \norm{E_f \setminus E'_f}
	3m^{2.5}\cdot 2 \norm{E_b}
	\Big)
  \\
  & = O(\delta' m^{5.5} \E(\Pi^*)).
\end{align*}
So we can afford to execute Arora's $(m,k)$-TSP algorithm with
parameter $\delta' = \frac{c \delta}{m^{5.5}}$ for suitable constant $c>1$.

\paragraph{Running time.}
The table $V$ has $n^{O(m)}$ entries, and computing each entry requires
consulting $n^{O(m)}$ other entries.
In addition, one has to invoke $n^{O(m)}$ times Arora's $k$-TSP algorithm
(modified as explained in Section~\ref{sec:prelims} to find $m$ tours),
and each of these is executed in time
$k^2 n(2^m\log n)^{(d/\delta')^{O(d)}}$.
Plugging our $\delta' = \frac{c \delta}{m^{5.5}}$ yields total running time
$n^{O(m)}(\log n)^{(md/\delta)^{O(d)}}
= n^{O(m)} 2^{(md/\delta)^{O(d)}}$,
which completes the proof of Theorem~\ref{THM: PTAS for (m,k)-TSP}.

\subsection{Auxiliary lemma}

Above we required the following lemma.

\begin{lemma}\label{lem:direction}
For every set of unit-length vectors $v_1, \ldots , v_m\in \R^d$,
there exists a unit-length $x\in \R^d$ (direction in space)
and signs $\sigma_1,\ldots,\sigma_m \in \set{\pm1}$
satisfying

\[
  \forall i\in[m],
  \quad
  \tuple{x, \sigma_iv_i} \geq \frac{1}{4m^{3/2}}
  \text{\quad and thus\quad }
  0 \leq \vecangle{x, \sigma_iv_i} \leq \frac{\bpi}{2}- \frac{1}{4 m^{3/2}} .
\]
The vector and signs may be found with constant probability in time $O(md)$.

Moreover, there exists a deterministic algorithm that runs in time $m^{O(d)}$ and
finds a vector and signs satisfying

\[
  \forall i\in[m],
  \quad
  \tuple{x, \sigma_iv_i} \geq \frac{1}{8m^{3/2}}
  \text{\quad and thus\quad }
  0 \leq \vecangle{x, \sigma_iv_i} \leq \frac{\bpi}{2}- \frac{1}{8 m^{3/2}} .
\]

\end{lemma}

%We actually prove that the inner product is at least $\frac{1}{8m\sqrt{d}}$,
%and arguing that without loss of generality $d\leq m$, the stated bound follows.

\begin{proof}
Let $y\in\R^d$ be a random vector where each entry is an iid Gaussian $N(0,1)$,
i.e., chosen from the distribution $y\sim N(0,I_d)$.
Now fix $i\in[m]$.
As the Gaussian distribution is 2-stable,
the inner product $\tuple{y,v_i}$ is distributed as
$N(0,\norm{y}^2=1)$.
For a standard Gaussian $g \sim N(0,1)$, since the probability density function of the distribution
has maximum value $\frac{1}{\sqrt{2\pi}}$, we have that

\[
  \forall \gamma\in(0,1),
  \quad
  \Pr_g \big( \abs{g} \in [0,\gamma] \big)
  \le 2\gamma/ \sqrt{2 \pi}
  < \gamma.
\]
Plugging in $\gamma=\tfrac{1}{2m}$, we have that
\[
  \Pr_y\big( \abs{\tuple{y,v_i}} \le \tfrac{1}{2m} \big)
  < \tfrac{1}{2m}.
\]
Now since $\E[\norm{y}^2]=d$,
Markov's inequality gives that $\Pr[\norm{y}^2\ge 4d] \le 1/4$.
Applying a union bound over the above $m$ events along with
the bound on $\norm{y}^2$,
we have that with probability at least $1/4$ all these events fail,
in which case $x=y/\norm{y}$ is a unit-length vector (in same direction as $y$)
satisfying
\[
  \forall i\in[m],
  \quad
  \abs{\tuple{x,v_i}}
  >   \tfrac{1}{2m \cdot \norm{y}}
  \ge \tfrac{1}{4 m\sqrt{d}} .
\]
Finally, for each $i\in[m]$ we can pick a sign $\sigma_i\in\set{\pm1}$
such that $\tuple{x,\sigma_i v_i} = \abs{\tuple{x,v_i}} \ge 0$.
We may assume that $d\leq m$,
as otherwise we can restrict attention to the span of the vectors,
and conclude the required
$\tuple{x,\sigma_i v_i} \ge \tfrac{1}{4 m^{3/2}}$.
To bound the angle,
let $\theta_i := \tfrac{\bpi}{2} - \vecangle{x, \sigma_iv_i}$
and then
$
  \tfrac{1}{4 m^{3/2}}
  \le \tuple{x,\sigma_i v_i}
  = \cos(\tfrac{\bpi}{2}-\theta_i)
  = \sin(\theta_i)
  \leq \theta_i
$,
where the last inequality relies on observing that
$\theta_i \in [0,\tfrac{\bpi}{2}]$.

The above immediately implies a randomize algorithm that runs in time $O(md)$.
The constant success probability is $1/4$,
but it can be amplified by only $O(1)$ independent repetitions.

For a deterministic algorithm, consider all the points of
a grid with spacing $\tfrac{1}{16 m^{3/2} \sqrt{d}}$ in $[-1,1]^d$.
The diameter of each grid cell is clearly $\tfrac{1}{16 m^{3/2}}$.
Let the point-set $R$ include these grid points,
and let $R'$ be the points of $R$ after normalizing each to have unit length.
Define $x$ to be the unit-length vector whose existence is proved above.
By construction, $x$ is within distance $\tfrac{1}{16 m^{3/2}}$
of some grid point $x' \in R$,
and this $x'$ is within distance $\tfrac{1}{16 m^{3/2}}$
of its normalized vector $x'' \in R'$.
It follows that for all $i$, we have
$|\tuple{v_i,x}-\tuple{v_i,x''}| = |\tuple{v_i,x-x''}| \le \norm{v_i}\cdot \norm{x-x''} \le \tfrac{1}{8 m^{3/2}}$,
and picking a sign $\sigma_i\in\set{\pm1}$ appropriately gives 
$\tuple{\sigma_iv_i,x''}
 = |\tuple{v_i,x''}| 
\geq |\tuple{v_i,x}| - \tfrac{1}{8 m^{3/2}}
\geq \tfrac{1}{8 m^{3/2}}$,
as required.
The algorithm locates this $x''$ or another suitable vector
by enumerating over all normalized grid points, whose number is
$|R'| = |R| 
\leq (2 \cdot 16 m^{3/2} \sqrt{d})^{d}
\leq (2 \cdot 16 m^2)^{d}$,
and for each of them computing the $m$ inner products,
which overall takes time $m^{O(d)}$. 
\end{proof}

\section{A PTAS for Orienteering}\label{sec:orient}

Having shown in Theorem \ref{THM: PTAS for (m,k)-TSP} how to compute a
$\delta$-excess-approximation to an optimal
$(m,k)$-TSP tour, we can use this algorithm as
a subroutine to solve the orienteering problem.
As in \cite{CH08}, we show how to reduce the orienteering
problem to $n^{O(1/\delta)}$ instances of the
$(O(1/\delta), k)$-TSP problem.

\begin{lemma}\label{lem:reduction}
A $(1-\delta)$-approximation to orienteering problem on $n$-point set $P$, a budget $\B$ and a starting point $s$,
can be computed by making
$n^{O(1/\delta)}$
queries to an $O(\delta)$-excess-approximation oracle for $(m,k)$-TSP,
with parameters
$m=O(1/\delta)$
and
$k=O(k_{\opt})$, where $k_{\opt}$ denotes the number of points visited by an optimal path.
\end{lemma}

Then Theorem \ref{THM: PTAS for Orienteering} follows from Lemma \ref{lem:reduction},
with oracle queries executed by the algorithm of Theorem~\ref{THM: PTAS for (m,k)-TSP}.

\begin{proof}

    Let $\pi^*$ be an optimal rooted orienteering path starting at $s$ of length at most $\B$ that visits $k_{\opt}$ points of $P$, let $\pi^*(i,j)=\langle p_i,\ldots,p_j \rangle$ be the portion of the path $\pi^*$ from $p_i$ to $p_j$,
    and let $\E(i,j)=\norm{\pi^*(i,j)}-\norm{p_i-p_j}$ be its excess.
    Set $m=\lfloor 1/\delta\rfloor$, and let $\alpha_i=\lceil (i-1)(k_{\opt}-1)/m\rceil +1$ for every $1\leq i\leq m+1$.
    By definition, we have $\alpha_1=1$ and $\alpha_{m+1}=k_{\opt}$. Furthermore, each subpath $\pi^*(\alpha_i,\alpha_{i+1})$ visits
    \[
      \alpha_{i+1}-\alpha_i-1= (\lceil i(k_{\opt}-1)/m\rceil +1)-(\lceil (i-1)(k_{\opt}-1)/m\rceil +1)-1\leq \lfloor(k_{\opt}-1)/m\rfloor
    \]
    points, excluding the endpoints $p_{\alpha_i}$ and $p_{\alpha_{i+1}}$.

    Consider the subpaths $\pi^*(\alpha_1,\alpha_{2}),\ldots,\pi^*(\alpha_{m},\alpha_{m+1})$ of $\pi^*$ and their
    respective excesses
    \[
	\E_1=\E(\alpha_1,\alpha_2), \ldots, \E_{ m} = \E(\alpha_{m},\alpha_{m +1}).
    \]
    Clearly, there exists an index $\nu$, $1\leq\nu \leq m$, such that $\E_{\nu}\geq \frac{1}{m}(\sum_{i=1}^{m}\E_i)$. By connecting the vertex $p_{\alpha_{\nu}}$ directly to the vertex $p_{\alpha_{\nu+1}}$ in $\pi^*$, we obtain a new path
    $\pi'=\langle p_1,p_2\ldots,p_{\alpha_{\nu}},p_{\alpha_{\nu+1}},p_{\alpha_{\nu+1}+1},\ldots,p_{k_{\opt}} \rangle$. Observe that $ \norm{\pi'}=\norm{\pi^*}-\E_{\nu}$, and
    as noted above,
    $\pi'$ visits at least
    \[
    k_{\opt}-(\alpha_{\nu+1}-\alpha_{\nu}-1)\geq k_{\opt}- \lfloor(k_{\opt}-1)/ m \rfloor\geq (1-1/m)k_{\opt}
    \]
    points of $P$.

%    Consider the $(m+1)$-point {\em skeleton} $\S'=\langle p_{\alpha_1},\ldots,p_{\alpha_{m+1}} \rangle$ of $\pi'$.
%    By the definition of $\E_i$, we have that
%    $\norm{\S'}=\norm{\pi^*}-\sum_{i=1}^{m}\E_i$.
%    Therefore, by the definition of $\E_{\pi'}$, we have that
%    \[
%	\E_{\pi'}\leq \norm{\pi'}-\norm{\S'} =(\norm{\pi^*}-\E_{\nu}) -(\norm{\pi^*}-\sum_{i=1}^{m}\E_i) = \sum_{i=1}^{m}\E_i - \E_{\nu}.
%    \]

    By applying an $(m,(1-1/m)k_{\opt})$-TSP oracle on the $m$ pairs $(s_i=p_{\alpha_i}, t_i=p_{\alpha_{i+1}})$ for every $i\in[m]$
    %(i.e. the pairs are determined by the points of the skeleton $\S'$)
    with accuracy parameter $1/m$, one can compute a path $\hat {\pi}$ that visits at least $(1-1/m)k_{\opt}\geq (1-\delta)k_{\opt}$ points of $P$, of length
    \begin{align*}
      \norm{\hat{\pi}}  & \leq \sum_{i=1}^{m}(\norm{\pi'(\alpha_i,\alpha_{i+1})}+\frac{1}{m}\E_i)\\
		                & \leq \norm{\pi'}+\E_{\nu} \\
                        & = \norm{\pi^*} \\
		                & \leq \B.
    \end{align*}
    %\begin{align*}
%      \norm{\hat{\pi}}  &\leq \norm{\pi'}+\frac{\E_{\pi'}}{m+1}	\\
%		& \leq (\norm{\pi^*}-\E_{\nu}) +\frac{1}{m+1}(\sum_{i=1}^{m}\E_i - \E_{\nu})\\
%                  &= \norm{\pi^*}+\frac{1}{m+1}(\sum_{i=1}^{m}\E_i - (m+2)\E_{\nu}) \\
%		& \leq \norm{\pi^*}\leq \B,
%    \end{align*}
    %since $\sum_{i=1}^{m}\E_i - (m +2)\E_{\nu}\leq 0$, as implied by $\E_{\nu}\geq \frac{1}{m}(\sum_{i=1}^{m}\E_i)$.

    As the value of $k_{\opt}$ is not known in advance,
    the algorithm tries all possible values of $k$ from $1$ to $n$,
    returning the maximum value $k'$ for which it finds a tour within budget $\B$
    (that is, the algorithm terminates at the failed attempt to find
    a tour visiting $k'+1$ points). As we proved above, $(1-\delta)k_{\opt} \leq k'\leq k_{\opt}$.
    In addition, since we do not know the optimal orienteering path $\pi^*$ in advance,
    we guess the $m=\lfloor 1/\delta \rfloor$ points $p_{\alpha_i}$, which gives
    $n^{O(1/\delta)}$ queries of $(m,k)$-TSP.
\end{proof}

\paragraph{Acknowledgements.} We thank Rajesh Chitnis for helpful discussions.

{
\ifprocs
\bibliographystyle{plainurl}
\else
\bibliographystyle{alphaurlinit}
\fi
\bibliography{havana}

\newcommand{\etalchar}[1]{$^{#1}$}
\begin{thebibliography}{BBCM04}

\bibitem[AMN98]{AMN98}
E.~M. Arkin, J.~S.~B. Mitchell, and G.~Narasimhan.
\newblock Resource-constrained geometric network optimization.
\newblock In {\em Proceedings of the Fourteenth Annual Symposium on
  Computational Geometry}, pages 307--316, 1998.
\newblock \href {http://dx.doi.org/10.1145/276884.276919}
  {\path{doi:10.1145/276884.276919}}.

\bibitem[Aro98]{Arora98}
S.~Arora.
\newblock Polynomial time approximation schemes for {E}uclidean traveling
  salesman and other geometric problems.
\newblock {\em J. {ACM}}, 45(5):753--782, 1998.
\newblock \href {http://dx.doi.org/10.1145/290179.290180}
  {\path{doi:10.1145/290179.290180}}.

\bibitem[ASV14]{archetti14}
C.~Archetti, M.~G. Speranza, and D.~Vigo.
\newblock Chapter 10: Vehicle routing problems with profits.
\newblock In {\em Vehicle Routing: Problems, Methods, and Applications, Second
  Edition}, pages 273--297. SIAM, 2014.
\newblock \href {http://dx.doi.org/10.1137/1.9781611973594.ch10}
  {\path{doi:10.1137/1.9781611973594.ch10}}.

\bibitem[Bal95]{Balas95}
E.~Balas.
\newblock The prize collecting traveling salesman problem: {II.} polyhedral
  results.
\newblock {\em Networks}, 25(4):199--216, 1995.
\newblock \href {http://dx.doi.org/10.1002/net.3230250406}
  {\path{doi:10.1002/net.3230250406}}.

\bibitem[BBCM04]{BBCM04}
N.~Bansal, A.~Blum, S.~Chawla, and A.~Meyerson.
\newblock Approximation algorithms for deadline-{TSP} and vehicle routing with
  time-windows.
\newblock In {\em Proceedings of the 36th Annual {ACM} Symposium on Theory of
  Computing}, pages 166--174, 2004.
\newblock \href {http://dx.doi.org/10.1145/1007352.1007385}
  {\path{doi:10.1145/1007352.1007385}}.

\bibitem[BCK{\etalchar{+}}07]{BCKLMM07}
A.~Blum, S.~Chawla, D.~R. Karger, T.~Lane, A.~Meyerson, and M.~Minkoff.
\newblock Approximation algorithms for orienteering and discounted-reward
  {TSP}.
\newblock {\em {SIAM} J. Comput.}, 37(2):653--670, 2007.
\newblock \href {http://dx.doi.org/10.1137/050645464}
  {\path{doi:10.1137/050645464}}.

\bibitem[CGRT03]{CGRT03}
K.~Chaudhuri, B.~Godfrey, S.~Rao, and K.~Talwar.
\newblock Paths, trees, and minimum latency tours.
\newblock In {\em 44th Symposium on Foundations of Computer Science {(FOCS}
  2003),}, pages 36--45, 2003.
\newblock \href {http://dx.doi.org/10.1109/SFCS.2003.1238179}
  {\path{doi:10.1109/SFCS.2003.1238179}}.

\bibitem[CH08]{CH08}
K.~Chen and S.~Har{-}Peled.
\newblock The {E}uclidean orienteering problem revisited.
\newblock {\em {SIAM} J. Comput.}, 38(1):385--397, 2008.
\newblock \href {http://dx.doi.org/10.1137/060667839}
  {\path{doi:10.1137/060667839}}.

\bibitem[CK04]{CK04}
C.~Chekuri and A.~Kumar.
\newblock Maximum coverage problem with group budget constraints and
  applications.
\newblock In {\em Approximation, Randomization, and Combinatorial Optimization,
  Algorithms and Techniques, 7th International Workshop on Approximation
  Algorithms for Combinatorial Optimization Problems, {APPROX} 2004, and 8th
  International Workshop on Randomization and Computation, {RANDOM} 2004,
  Cambridge,}, pages 72--83, 2004.
\newblock \href {http://dx.doi.org/10.1007/978-3-540-27821-4\_7}
  {\path{doi:10.1007/978-3-540-27821-4\_7}}.

\bibitem[CKP12]{CKP12}
C.~Chekuri, N.~Korula, and M.~P{\'{a}}l.
\newblock Improved algorithms for orienteering and related problems.
\newblock {\em {ACM} Trans. Algorithms}, 8(3):23:1--23:27, 2012.
\newblock \href {http://dx.doi.org/10.1145/2229163.2229167}
  {\path{doi:10.1145/2229163.2229167}}.

\bibitem[CP05]{CP05}
C.~Chekuri and M.~P{\'{a}}l.
\newblock A recursive greedy algorithm for walks in directed graphs.
\newblock In {\em 46th Annual {IEEE} Symposium on Foundations of Computer
  Science {(FOCS} 2005), 23-25 October 2005, Pittsburgh, PA, USA, Proceedings},
  pages 245--253, 2005.
\newblock \href {http://dx.doi.org/10.1109/SFCS.2005.9}
  {\path{doi:10.1109/SFCS.2005.9}}.

\bibitem[FS17]{FS17}
Z.~Friggstad and C.~Swamy.
\newblock Compact, provably-good lps for orienteering and regret-bounded
  vehicle routing.
\newblock In {\em Integer Programming and Combinatorial Optimization - 19th
  International Conference, {IPCO} 2017, Proceedings}, volume 10328 of {\em
  Lecture Notes in Computer Science}, pages 199--211. Springer, 2017.
\newblock \href {http://dx.doi.org/10.1007/978-3-319-59250-3\_17}
  {\path{doi:10.1007/978-3-319-59250-3\_17}}.

\bibitem[GGJ76]{GGJ76}
M.~R. Garey, R.~L. Graham, and D.~S. Johnson.
\newblock Some {NP}-complete geometric problems.
\newblock In {\em Proceedings of the 8th Annual {ACM} Symposium on Theory of
  Computing}, pages 10--22, 1976.
\newblock \href {http://dx.doi.org/10.1145/800113.803626}
  {\path{doi:10.1145/800113.803626}}.

\bibitem[GKNR15]{GKNR15}
A.~Gupta, R.~Krishnaswamy, V.~Nagarajan, and R.~Ravi.
\newblock Running errands in time: Approximation algorithms for stochastic
  orienteering.
\newblock {\em Math. Oper. Res.}, 40(1):56--79, 2015.
\newblock \href {http://dx.doi.org/10.1287/moor.2014.0656}
  {\path{doi:10.1287/moor.2014.0656}}.

\bibitem[GLV87]{GLV87}
B.~L. Golden, L.~Levy, and R.~Vohra.
\newblock The orienteering problem.
\newblock {\em Naval Research Logistics (NRL)}, 34(3):307--318, 1987.
\newblock \href
  {http://dx.doi.org/10.1002/1520-6750(198706)34:3<307::AID-NAV3220340302>3.0.CO;2-D}
  {\path{doi:10.1002/1520-6750(198706)34:3<307::AID-NAV3220340302>3.0.CO;2-D}}.

\bibitem[GLV16]{gunawan16}
A.~Gunawan, H.~C. Lau, and P.~Vansteenwegen.
\newblock Orienteering problem: {A} survey of recent variants, solution
  approaches and applications.
\newblock {\em European Journal of Operational Research}, 255(2):315--332,
  2016.
\newblock \href {http://dx.doi.org/10.1016/j.ejor.2016.04.059}
  {\path{doi:10.1016/j.ejor.2016.04.059}}.

\bibitem[Mit99]{Mitchell99}
J.~S.~B. Mitchell.
\newblock Guillotine subdivisions approximate polygonal subdivisions: {A}
  simple polynomial-time approximation scheme for geometric {TSP}, $k$-{MST},
  and related problems.
\newblock {\em {SIAM} J. Comput.}, 28(4):1298--1309, 1999.
\newblock \href {http://dx.doi.org/10.1137/S0097539796309764}
  {\path{doi:10.1137/S0097539796309764}}.

\bibitem[Pap77]{Papadimitriou77}
C.~H. Papadimitriou.
\newblock The {E}uclidean traveling salesman problem is {NP}-complete.
\newblock {\em Theor. Comput. Sci.}, 4(3):237--244, 1977.
\newblock \href {http://dx.doi.org/10.1016/0304-3975(77)90012-3}
  {\path{doi:10.1016/0304-3975(77)90012-3}}.

\bibitem[Tre00]{Trevisan00}
L.~Trevisan.
\newblock When hamming meets {E}uclid: The approximability of geometric {TSP}
  and steiner tree.
\newblock {\em {SIAM} J. Comput.}, 30(2):475--485, 2000.
\newblock \href {http://dx.doi.org/10.1137/S0097539799352735}
  {\path{doi:10.1137/S0097539799352735}}.

\bibitem[TV02]{TV02}
P.~Toth and D.~Vigo, editors.
\newblock {\em The Vehicle Routing Problem}, volume~9 of {\em {SIAM} monographs
  on discrete mathematics and applications}.
\newblock {SIAM}, 2002.
\newblock \href {http://dx.doi.org/10.1137/1.9780898718515}
  {\path{doi:10.1137/1.9780898718515}}.

\end{thebibliography}
}

%\appendix

%\section{Appendix}

\end{document}